\documentclass[a4paper,11pt]{article}
\usepackage{amsmath,amssymb}
\usepackage{amsfonts}
\usepackage{eucal}
\usepackage{amsthm}
\numberwithin{equation}{section}

\newcommand{\bigpare}[1]{\bigl(#1\bigr)}
\newcommand{\biggpare}[1]{\biggl(#1\biggr)}
\newcommand{\Bigpare}[1]{\Bigl(#1\Bigr)}

\newcommand{\bigbrac}[1]{\bigl[#1\bigr]}

\newcommand{\biggbrac}[1]{\biggl[#1\biggr]}

\newcommand{\bigset}[2]{\bigl\{#1\bigm|#2\bigr\}}

\newcommand{\bignorm}[1]{\bigl\| #1 \bigr\|}
\newcommand{\Bignorm}[1]{\Bigl\| #1 \Bigr\|}

\newcommand{\bigabs}[1]{\bigl| #1 \bigr|}

\newcommand{\biggabs}[1]{\biggl| #1 \biggr|}

\newcommand{\jap}[1]{\langle #1 \rangle}

\def\a{\alpha}
\def\b{\beta}

\def\d{\delta}

\def\f{\varphi}

\def\m{\mu}

\def\x{\xi}
\def\y{\eta}

\newcommand{\F}{\Phi}
\newcommand{\G}{\Psi}
\renewcommand{\L}{\Lambda}
\renewcommand{\O}{\Omega}

\def\re{\mathbb{R}}

\def\ze{\mathbb{Z}}

\def\pa{\partial}
\def\torus{\mathbb{T}}

\newcommand{\supp}{\text{{\rm supp}\;}}

\DeclareMathOperator*{\slim}{s-lim}

\newcommand{\Ran}{\text{\rm Ran\;}}

\newtheorem{thm}{Theorem}[section]
\newtheorem{lem}[thm]{Lemma}
\newtheorem{prop}[thm]{Proposition}

\theoremstyle{definition}
\newtheorem{defn}{Definition}[section]

\theoremstyle{remark}
\newtheorem{rem}{Remark}

\begin{document}

\title{Modified wave operators for discrete Schr\"odinger operators with long-range perturbations}
\author{Shu Nakamura\footnote{%
Graduate School of Mathematical Sciences, University of Tokyo, 
3-8-1, Komaba, Meguro, Tokyo, 153-8914 Japan.
E-mail: {\tt shu@ms.u-tokyo.ac.jp}.
The research was partially supported by JSPS Grant Kiban (A) 21244008.}}

\maketitle

\begin{abstract}
We consider the scattering theory for discrete Schr\"odinger operators on $\ze^d$ with long-range potentials. 
We prove the existence of modified wave operators constructed in terms of solutions of a Hamilton-Jacobi 
equation on the torus $\torus^d$. 
\end{abstract}

\section{Introduction}

We consider the discrete Schr\"odinger operator 
\[
Hu[n]= -\frac12\triangle u[n] +V[n] u[n], \quad n\in\ze^d,
\]
for $u\in \mathcal{H}=\ell^2(\ze^d)$, where 
\[
\triangle u[n] =\sum_{|m-n|=1} u[m],
\]
and $V$ is a real-valued function on $\ze^d$. 
We denote discrete variables using the square braces $[\cdot]$, and continuous variables using the 
round braces $(\cdot)$. If $V$ is bounded,  $H$ is a bounded self-adjoint operator on $\ell^2(\ze^d)$. 

The discrete Schr\"odinger operator $H$ has many common properties as the continuous Schr\"odinger 
operator on $\re^d$. For example, if $V$ is short-range type, i.e., 
\[
|V[n]|\leq C(1+|n|)^{-\m}, \quad n\in\ze^d, 
\]
with some $\m>1$ and $C>0$, then the scattering theory is constructed in the standard way. Namely, 
by setting $H_0=-\frac12\triangle$, we can show the wave operators 
\[
W_\pm =\slim_{t\to\pm\infty} e^{itH} e^{-itH_0}
\]
exist; it is an isometry into $\mathcal{H}_{ac}(H)$, the absolutely continuous subspace of $H$; 
the intertwining property: $HW_\pm=W_\pm H_0$ holds; moreover 
they are asymptotically complete: $\Ran W_\pm =\mathcal{H}_{ac}(H)$ 
(see, e.g, Boutet de Monvel, Sahbani \cite{BS}, Isozaki, Korotyaev \cite{IK} and references therein). 

We will consider the long-range case, i.e., when $0<\m\leq 1$. If $V$ is long-range type, the wave 
operators do not exist in general, and we need to introduce {\em modified}\/ wave operators. 
In order to state our main result, we introduce several notations. 

Following the standard notation, we denote $W\in S^m(\re^d)$, $m\in\re$, if $W\in C^\infty(\re^d)$ 
and for any multi-index $\a\in\ze_+^d$ there is $C_\a>0$ such that 
\[
\bigabs{\pa_x^\a W(x)}\leq C_\a \jap{x}^{m-|\a|},\quad x\in\re^d, 
\]
where $\jap{x}=\sqrt{1+|x|^2}$. 
We note that when we consider scattering theory for (continuous) Schr\"odinger operators, 
$V\in S^{-\m}(\re^d)$
with $\m>0$ is a standard assumption. A natural analogue for the discrete case is the following. 
We denote 
\[
\tilde\partial_j u[n] = u[n] -u[n-e_j], \quad n\in\ze^d, j=1,\dots,d, 
\]
where $\{e_j\}$ is the standard orthonormal basis of $\re^d$, and $u[\cdot]$ is a function on $\ze^d$. 
We denote $\tilde\pa^\a=\prod_{j=1}^d \tilde\pa_j^{\a_j}$ for $\a\in\ze_+^d$ as usual. 

\begin{defn} Let $V$ be a function on $\ze^d$, and let $m\in\re$. 
We denote $V\in S^{m}(\ze^d)$ if for any $\a\in\ze_+^d$ there is $C_\a>0$ such that 
\[
\bigabs{\tilde \pa^\a V[n]}\leq C_\a \jap{n}^{m-|\a|}, \quad n\in\ze^d.
\]
\end{defn}

We suppose $V\in S^{-\m}(\ze^d)$ with $\m>0$. Then, the essential spectrum on $H$ is $[-d,d]$, 
and we are interested in the structure of the spectrum of $H$ in $[-d,d]$. 

In the next section, we show that we can extend $V\in S^m(\ze^d)$ to an element in $S^m(\re^d)$. 
We denote $\torus =\re/(2\pi\ze)$. 
Then we may consider the classical mechanics generated by $p(x,\x)$, the symbol of $H$ on 
$T^*\torus^d\cong \re^d\times\torus^d$, that is 
\[
p(x,\x)= \sum_{j=1}^d \cos(\x_j) +V(x)\quad (x,\x)\in \re^d\times\torus^d.
\]
We denote the set of threshold energies by 
$\mathcal{T}=\{-d,-d+2,\dots, d-2,d\}$. 
For a given interval $I\Subset [-d,d]\setminus\mathcal{T}$, 
we can find solutions to the Hamilton-Jacobi equation on $\torus^d$:
\[
\frac{\pa}{\pa t}\F_\pm(t,\x) = p(\pa_\x\F_\pm(t,\x),\x), \quad \x\in\torus^d,  \ \pm t \geq 0, 
\ p_0(\x)\in I, 
\]
such that 
\[
|\F_\pm(t,\x)-tp_0(\x)| = O(|t|^{1-\m})
\]
as $t\to\pm\infty$, where $p_0(\x)= \sum_{j=1}^d \cos(\x_j)$.  

We denote the discrete Fourier transform by $F$, i.e., 
\[
Fu(\x) := (2\pi)^{-d/2} \sum_{n\in\ze^d} e^{-in\cdot \x} u[n] , \quad \x\in\torus^d=(\re/2\pi\ze)^d. 
\]
We note $F$ is a unitary map from $\ell^2(\ze^d)$ to $L^2(\torus^d)$. 
For a function $f(\x)$ on $\torus^d$, we denote 
\[
f(D_x)u = F^*(f(\x)(Fu)(\x)), \quad  u\in \ell^2(\ze^d). 
\]
Using these, we can state our main result: 

\begin{thm}
Suppose $V\in S^{-\m}(\ze^d)$ with $\m>0$. 
Let $I\Subset [-d,d]\setminus\mathcal{T}$, and and let $\F_\pm(t,\x)$ be as above (or as constructed in Section~4). 
Then the modified wave operators 
\[
W_\pm^\F(I)= \slim_{t\to\pm\infty} e^{itH} e^{-i\F(t,D_x)} E_I(H_0)
\]
exist and they are isometry from $\ell^2(\ze^d)$ to $\mathcal{H}_{ac}(H)$. 
Moreover, the intertwining property: 
\[
H W_\pm^\F(I) = W_\pm^\F(I) H_0
\]
holds. 
\end{thm}

We expect the asymptotic completeness: $\Ran W_\pm^\F(I) =\mathcal{H}_{ac}(H)$ holds, 
though we do not discuss it in this paper. 

The long-range scattering theory for Schr\"odinger equation has long history, 
starting from the pioneering work by Dollard \cite{Do}. We refer Reed-Simon \cite{RS} \S XI.9, 
Yafaev \cite{Ya} Chapter~10, Derezinski-G\'erard \cite{DG} \S4.7, and references therein. 
We follow the argument of H\"ormander \cite{Ho} in this paper to construct modified 
wave operators. 

We prepare a simple extension lemma in Section~2. We discuss the classical mechanics 
on $\torus^d$ in Section~3, and solutions to the Hamilton-Jacobi equation are constructed in 
Section~4. We prove Theorem~1.1 in Section~5. 

\section{Prelininaries}

Here we construct an extension of $V\in S^m(\ze^d)$ to a smooth function $\tilde V(x)$ on $\re^d$. 

\begin{lem}
Suppose $V[\cdot]\in S^m(\ze^d)$. Then there is $\tilde V\in S^{m}(\re^d)$ such that 
it is real-valued and $\tilde V(n)=V[n]$ for any $n\in\ze^d$. 
\end{lem}

\begin{proof}
We interpolate $V[n]$ using a {\it window function}\/ (see, e.g., Oppenheim-Schafer-Buck \cite{OS} \S7.2). 
Let $\hat\chi_0\in C_0^\infty(\re^d)$ such that 
\begin{enumerate}
\renewcommand{\labelenumi}{(\roman{enumi}) }
\item $\sum_{n\in\ze^d} \hat\chi_0(\x+2\pi n)=1$, $\x\in\re^d$. 
\item $\supp \hat\chi_0(\x)\subset \bigbrac{-\frac32\pi,\frac32\pi}^d$. 
\item $\hat\chi_0$ is even, and $\hat\chi_0(\x)\geq 0$, $\x\in\re^d$. 
\end{enumerate}
Such $\hat \chi_0$ is called a window function. Let $\chi_0$ be the inverse Fourier transform 
of $\hat\chi_0$, and we set
\[
\tilde V(x)=(2\pi)^{-d/2} \sum_{n\in\ze^d} \chi_0(x-n) V[n], \quad x\in\re^d.
\]
Since $\chi_0\in\mathcal{S}(\re^d)$ is real-valued by (iii)  and $V[n]=O(\jap{n}^m)$, it is easy to see 
$\tilde V$ is a smooth real-valued function on $\re^d$. On the other hand, 
if $k\in\ze^d$, then 
\begin{align*}
\chi_0(k) &= (2\pi)^{-d/2} \int e^{ik\cdot x} \chi_0(\x)d\x\\
&= (2\pi)^{-d/2} \int_{[-\pi,\pi]^d} e^{ik\cdot x}d\x 
=(2\pi)^{d/2} \d_{k0}, 
\end{align*}
by virtue of the property (i). This implies $\tilde V(n)=V[n]$ for $n\in \ze^d$. 

We now show $\tilde V\in S^{m}(\re^d)$. It is easy to see $|\tilde V(x)|\leq C\jap{x}^{m}$ 
since $\chi_0\in \mathcal{S}(\re^d)$ and $V[n]=O(\jap{n}^{m})$ as $|n|\to\infty$. 

We note $Fu=\mathcal{F}\bigpare{\sum_n \d(x-n)u[n]}$, 
where $F$ is the discrete Fourier transform and $\mathcal{F}$ is the standard Fourier transform. 
We set 
\[
\hat V(\x)= FV(\x)\ \in \mathcal{E}'(\torus^d).
\]
We note the Fourier transform of $\tilde V(\cdot)$ is given by 
$(\mathcal{F}\tilde V)(\x)=\hat\chi_0(\x)\hat V(\x)$. We also note 
\begin{align*}
& F(\tilde\pa_j V)(\x)= (1-e^{-i\x_j})\hat V(\x) =2i e^{-i\x_j/2}\sin(\x_j/2)\hat V(\x), \\
& \mathcal{F}(\pa_j \tilde V)(\x) =i\x_j \hat\chi_0(\x)\hat V(\x).
\end{align*}
Since $\hat \chi_0$ is supported in $\bigbrac{-\frac32\pi,\frac32\pi}^d$, we may write 
\begin{align*}
\mathcal{F}(\pa_j\tilde V)(\x) &= \biggbrac{\frac{e^{i\x_j/2}\x_j}{2\sin(\x_j/2)}\hat\chi_0(\x)}
F(\tilde\pa_j V)(\x) \\
&= \hat\chi_j(\x) F(\tilde\pa_j V)(\x),
\end{align*}
with $\hat \chi_j\in C_0^\infty(\re^d)$. Thus we learn 
\[
\pa_j \tilde V(x) =(2\pi)^{-d/2} \sum_{n\in\ze^d} \chi_j(x-n)\tilde\pa_j V[n],
\]
where $\chi_j=\mathcal{F}^*\hat \chi_j$. Since $\tilde \pa_j V[n]=O(\jap{n}^{m-1})$ as $|n|\to\infty$, 
we have $\pa_j \tilde V(x) = O(\jap{x}^{m-1})$ as $|x|\to\infty$. 
Repeating this procedure, we conclude $\tilde V\in S^{m}(\re^d)$. 
\end{proof}

\section{Classical mechanics}

In the following we suppose $V\in S^{-\m}(\ze^d)$ with $\m>0$. 
Let $\tilde V(x)\in S^{m}(\re^d)$ be an extension of $V[n]$, and we write $\tilde V(x)=V(x)$ for simplicity. 
The existence of such $\tilde V$ is shown in Lemma~2.1, but it is not unique, and  
we may choose different extension. For example, if $V[n]=c\jap{n}^{-\m}$, then 
it is natural to choose $\tilde V(x)=c \jap{x}^{-\m}$, which is different from the extension in Lemma~2.1. 

We now construct a classical mechanics on $T^*\torus^d$ corresponding to 
the discrete Schr\"odinger operator. We write  the symbols of $H$ and $H_0$ by 
\[
p(x,\x)= p_0(\x)+ V(x), \quad p_0(\x)=\sum_{j=1}^d \cos(\x_j)
\]
for $(x,\x)\in\re^d\times\torus^d\cong T^*\torus^d$, respectively. 
It is easy to see 
\[
H_0u=F^* \bigpare{p_0(\x) (Fu)(\x)}\quad \text{for }\  u\in\ell^2(\ze^d).
\]

We consider the solutions to the Hamilton equation 
\[
\frac{d}{dt} x_j(t) =\frac{\pa p}{\pa \x_j}(x,\x) =\sin(\x_j), \quad
\frac{d}{dt} \x_j(t) = -\frac{\pa p}{\pa x_j}(x,\x) =-\frac{\pa V}{\pa x_j}(x),
\]
with an initial condition $(x(0),\x(0))=(x_0,\x_0)\in \re^d\times\torus^d$. 
In this section, we study long-time behavior of the solution. 

Let $\mathcal{T} = \{-d.-d+2,\dots, d-2, d\}$ 
be the set of threshold energies for $p_0(\x)$. 
Note 
\begin{align*}
\sin(\x_j)=0, \ j=1,\dots, d \ &\Longleftrightarrow \ 
\x_j\in\pi\ze, \ j=1,\dots, d \\
&\Longrightarrow \ p_0(\x)\in\mathcal{T}. 
\end{align*}
We denote the velocity by 
\[
v(\x) = (\sin(\x_1),\cdots, \sin(\x_d))\in\re^d, \quad \x\in\torus^d.
\]
We compute 
\begin{align}
\frac{d}{dt} |x(t)|^2 &= 2\sum_{j=1}^d v_j(\x(t))x_j(t), \nonumber \\
\frac{d^2}{dt^2}|x(t)|^2 &= 2\sum_{j=1}^d v_j(\x(t))^2 -2\sum_{j=1}^d x_j(t)\frac{\pa V}{\pa x_j}(x(t)).
\end{align}
We write 
\[
k(\x)=|v(\x)|^2 =\sum_{j=1}^d (\sin(\x_j))^2.
\]
We note $k(\x)>0$ if $p_0(\x)\notin \mathcal{T}$. 
For $I\subset (-d,d)$, we set
\[
\O_\pm(I,R)=\bigset{(x,\x)\in\re^d\times\torus^d}{p(x,\x)\in I, |x|\geq R, \pm x\cdot v(\x)\geq 0}.
\]

\begin{prop}
Let $I\Subset [-d,d]\setminus \mathcal{T}$ and $\O_\pm(I,R)$ be as above. 
Then there are $R_0>0$ and $\d>0$ such that 
\[
|x(t)|\geq \sqrt{R^2+\d t^2} \quad \text{for }\pm t\geq 0,
\]
if the initial condition $(x_0,\x_0)\in\O_\pm(I,R)$ with $R\geq R_0$. 
\end{prop}

\begin{proof} We choose $\d>0$ so that 
\[
(I+[-\d,\d])\cap\mathcal{T}=\emptyset, \quad 
\inf\bigset{k(\x)}{p_0(\x)\in I+[-\d,\d]}\geq 2\d.
\]
We also choose $R_0>0$ so that 
\[
|V(x)|\leq \d, \quad |x\cdot\nabla V(x)|\leq \d \quad \text{if } |x|\geq R_0.
\]
If $|x|\geq R_0$ and $p(x,\x)\in I$, then 
$p_0(x,\x)\in I+[-\d,\d]$, and hence $k(\x)\geq 2\d$. By virtue of (3.1), we then learn 
\[
\frac{d^2}{dt^2}|x(t)|^2 = 2k(\x) -2x\cdot\nabla V(x) \geq 2\d
\]
for such $(x,\x)$. If $(x_0,\x_0)\in\O_\pm(I,R)$ then $\pm\frac{d}{dt}|x(t)|^2 \geq 0$ at 
$t=0$, and hence 
\[
|x(t)|^2\geq |x_0|^2 +\d t^2 \quad \text{for } \pm t\geq 0.
\]
\end{proof}

This implies, in particular, 
\[
|x(t)|\geq \frac{1}{\sqrt{2}} (R+\sqrt{\d}|t|)
\quad \text{for }\pm t\geq 0.
\]
Once this estimate is established, following estimates are proved exactly same way as in the Euclidean 
space case (see, e.g., \cite{Ho}, \cite{CKS}). 
We denote the solution to the Hamilton equation with the initial condition $(y,\y)$ by 
\[
x(t)=x(t,y,\y), \quad \x(t)=\x(t,y,\y).
\]

\begin{prop}
Let $(x_0,\x_0)\in\O_\pm(I,R)$ and $x(t)=x(t,x_0,\x_0)$, $\x(t)=\x(t,x_0,\x_0)$. Then 
\[
\x_\pm =\lim_{t\to\pm\infty} \x(t)
\]
exists. Moreover, 
\begin{align}
&|\x(t)-\x_\pm|\leq C\jap{t}^{-\m}, \quad \pm t\geq 0, \\
& |x(t)-tv(\x_\pm)| \leq C \jap{t}^{1-\m}, \quad \pm t\geq 0. 
\end{align}
\end{prop}

We note, 
\[
|x(t)-t v(\x(t))|\leq C\jap{t}^{1-\m}
\]
is proved at first by the Hamilton equation, and combining this with (3.2) we obtain (3.3). 

\begin{prop}
There is $C>0$ such that 
\[
\biggabs{\frac{\pa}{\pa y} \x(t,y,\y)} \leq C R^{-1-\m}, 
\quad \biggabs{\frac{\pa}{\pa \y} \x(t,y,\y)} \leq C R^{-\m}
\]
for $R\geq R_0$, $(y,\y)\in \O_\pm(I,R)$, $\pm t\geq 0$. 
Moreover, for any $\a,\b\in\ze_+^d$, there is $C_{\a\b}>0$ such that 
\[
\biggabs{\pa_y^\a\pa_\y^\b (x(t,y,\y)-y)} \leq C_{\a\b} |t|, 
\quad \biggabs{\pa_y^\a\pa_\y^\b \x(t,y,\y)} \leq C_{\a\b} 
\]
for any $(y,\y)\in\O_\pm(I,R)$, $\pm t\geq 0$. 
\end{prop}

\section{Construction of solutions to the Hamilton-Jacobi equation}

As before, we set $p(x,\x)= \sum\cos(\x_j) +V(x)$, 
and construct solutions to the Hamilton-Jacobi equation 
\begin{equation}\label{HJ}
\frac{\pa}{\pa t}\F_\pm(t,\x) = p(\pa_\x\F_\pm(t,\x),\x), 
\quad \pm t\geq 0
\end{equation}
for $\x\in \bigset{\x}{p_0(\x)\in I}$, where $I\Subset (-d,d)\setminus\mathcal{T}$. We suppose 
$I+[-\d,\d]\subset (-d,d)\setminus\mathcal{T}$ with some $\d>0$. 

The characteristic equations for \eqref{HJ} is 
\[
\x'=-\nabla V(x), \quad x'=v(\x), \quad u'= p(x,\x)-x\cdot\nabla V(x), 
\]
where $x=\frac{\pa}{\pa\x}\F_\pm$ and $u=\F_\pm$ on the characteristic curves. 
The first two equations are the Hamilton equation, and we solve the equation with the initial condition:
\[
u(0) = \pm R_1 p_0(\y), \quad \x(0)=\y\in \torus^d, 
\]
with sufficiently large $R_1>0$. Then 
\[
x(0)=\frac{\pa u}{\pa \x}(0)=\pm R_1 v(\y),
\]
and we choose $R_1$ so large that 
\[
R:=R_1\cdot \inf\bigset{|v(\y)|}{p_0(\y)\in I+[-\d,\d]}
\]
satisfies the condition of Proposition~3.1, and $CR^{-\m}\ll 1$ in Proposition~3.3. 
We denote
\[
\L_t\ :\ \y\mapsto \x(t,\pm R_1v(\y),\y), \quad \pm t\geq 0. 
\]
Then $\L_t$ is locally diffeomorphic, and the derivatives are uniformly bounded in $t$. If $R_1$ is 
sufficiently large, we can easily show that $\L_t^{-1}$ is well-defined on $\bigset{\x}{p_0(\x)\in I}$, 
and the image is contained in $\bigset{\x}{p_0(\x)\in I+[-\d,\d]}$. Thus the solution to \eqref{HJ} is 
given by 
\[
\F_\pm(t,\x) =u\circ \L_t^{-1}, \quad 
u(t,\y)=R_1p_0(\y) +\int_0^t\bigpare{p(x(s),\x(s))-x(s)\cdot \nabla V(x(s))}ds.
\]
Moreover, by the construction, we have 
\[
\pa_\x \F_\pm(t,\x) = x(t,Rv(\y),\y), \quad \y =\L_t^{-1}(\x).
\]
Thus properties of $\pa_\x^\a\F_\pm$, $\a\neq 0$, follow from properties of 
$x(t,y,\y)$. In particular, we have 
\[
\bigabs{\pa_\x^\a\F_\pm(t,\x)}\leq C\jap{t}, \quad 
\pm t\geq 0, \ \x\in \bigset{\x}{p_0(\x)\in I}, 
\]
if $\a\neq 0$. Also, by the definition of $u$, we learn 
\begin{align*}
&\bigabs{\F_\pm(t,\x) -t p_0(\x)}\leq C\jap{t}^{1-\m}, \quad \pm t\geq 0, \\
&\bigabs{\pa_\x^\a(\pa_\x\F_\pm(t,\x)-tv(\x))}\leq C\jap{t}^{1-\m}, \quad \pm t>0,
\end{align*}
for any $\a\in \ze_+^d$. 

\section{Existence of modified wave operators}

Ler $\F_\pm(t,\x)$, etc., be as in Section~3. 
We fix $I\Subset (-d,d)\setminus \mathcal{T}$. 
We show 

\begin{thm} Suppose $V\in S^{-\m}(\ze^d)$ with $\m>0$. Then the modified wave operators 
\[
W^\F_\pm(I)=\slim_{t\to\pm\infty} e^{itH} e^{-i\F_\pm(t,D_x)} E_I(H_0)
\]
exists. 
\end{thm}

Generally, we follow the argument of H\"ormander \cite{Ho} to prove Theorem~5.1. 
We denote
\[
D(I) = \bigset{\x\in\torus^d}{p_0(\x)\in I, \cos(\x_j)\neq 0, 
j=1,\dots, d}.
\]
Then $C_0^\infty(D(I))$ is dense in 
$\mathrm{Ran}\, E_I(H_0)$. Thus it suffices to show the existence 
of $W_\pm(I)\f$ for $\f\in C_0^\infty(D(I))$. Moreover, by the partition of 
unity, we may suppose $\f$ is supported in an arbitrarily small neighborhood 
of a point $\x_0\in D(I)$. We compute 
\begin{align*}
\f(t,x)& := e^{-i\F_\pm(t,D_x)}\f(x)\\
&= (2\pi)^{-d/2} \int_{\torus^d} e^{i(x\cdot\x-\F_\pm(t,\x))}
\hat \f(\x)d\x
\end{align*}
for $x\in\ze^d$, where $\hat \f = F\f$. We write 
\[
x\cdot\x -\F_\pm(t,\x) =t \biggpare{\frac{x}{t}\cdot\x-\frac{1}{t}\F_\pm(t,\x)}
\]
and consider $t$ as the large parameter when we apply the stationary phase method. 
The stationary phase point is then given by 
\[
\frac{x}{t}=\frac{1}{t}\pa_\x\F_\pm(t,\x) = v(\x) +O(\jap{t}^{-\m}), 
\quad \pm t>0.
\]
From this, we also learn that the determinant of the Hessian of $\frac1t\F_\pm(t,\x)$ is 
$\prod \cos(\x_j) +O(\jap{t}^{-\m})$, and hence its absolute value is 
uniformly bounded from below by a positive constant on the support of $\hat\f$ 
with large $t$. Also the derivatives of $\frac1t \F_\pm(t,\x)$ in $\x$ are uniformly 
bounded in $t$ on the support of $\hat \f$. 

Let $D'\subset D(I)$ be a small neighborhood of $\supp\hat\f$. We denote
\[
G_t\ :\ \x\mapsto \pa_\x \F_\pm(t,\x).
\]
We may suppose $G_t$ is diffeomorphis on $D'$, and we note 
$\mathrm{vol}(G_t(D'))=O(\jap{t}^d)$ as $t\to\pm\infty$. 

By the non stationary phase method, we first learn 
\begin{equation}\label{NSP}
|\f(t,x)|\leq C_N \jap{|x|+|t|}^{-N} \quad \text{if }x\notin G_t(D'),
\end{equation}
with any $N$. On the other hand, by the stationary phase method, we have 
\begin{equation}\label{SP}
\f(t,x)= t^{-d/2} J(t,\y) \hat \f(\y) +O(\jap{t}^{-d/2-1})
\quad \text{for } x\in G_t(D'), 
\end{equation}
where $\y= G_t^{-1}(x)$, $J(t,\y)$ is a uniformly bounded function of $t,\y$, 
depending only on $\F_\pm(t,\y)$, $\pa_\x\F_\pm(t,\y)$ and 
$\pa_\x\pa_\x\F_\pm(t,\y)$. In particular, $J(t,\y)$ is independent of $\f$. 
(see, e.g., \cite{Ho2}, \S7.7). 

Now we estimate 
\[
\frac{d}{dt} \Bigpare{ e^{itH} e^{-i\F_\pm(t,D_x)}\f} 
= i e^{itH} (H-\pa_t\F_\pm(t,D_x)) e^{-i\F_\pm(t,D_x)}\f
\]
and apply the Cook-Kuroda method. We note, by \eqref{HJ}, it suffices to show 
\begin{align}
&\int_0^{\pm\infty} \Bignorm{\bigpare{H-p(\pa_\x\F_\pm(t,D_x),D_x)}
e^{-i\F_\pm(t,D_x)}\f} dt \nonumber\\
&\qquad = \int_0^{\pm\infty} \Bignorm{\bigpare{V(x)-V(\pa_\x\F_\pm(t,D_x))}
e^{-i\F_\pm(t,D_x)}\f} dt <\infty. \label{CK}
\end{align}
By \eqref{NSP}, it is easy to see 
\[
\int_0^{\pm\infty} \Bignorm{\chi_{G_t(D'))^c}(x) 
\bigpare{V(x)-V(\pa_\x\F_\pm(t,D_x))}\f(t,x)} dt <\infty.
\]
On the other hand, by \eqref{SP}, on $G_t(D')$ we have 
\begin{align*}
& V(x) \f(t,x) = t^{-d/2} V(x) J(t,\y) \hat\f(\y) + O(\jap{t}^{-d/2-1-\m}), \\
&V(\pa_\x\F_\pm(t,D_x))\f(t,\x) = t^{-d/2} J(t,\x) V(\pa_\x\F_\pm(t,\y))\hat \f(y)
+O(\jap{t}^{-d/2-1-\m}),
\end{align*}
where $x=\pa_\x\F_\pm(t,\y)$, and the leading terms of these coincide.
Thus we learn 
\begin{align*}
&\int_0^{\pm\infty} \Bignorm{\chi_{G_t(D'))}(x) 
\bigpare{V(x)-V(\pa_\x\F_\pm(t,D_x))}\f(t,x)} dt \\
&\quad \leq \int_0^{\pm\infty} C\jap{t}^{-d/2-1-\m} 
\bignorm{\chi_{G_t(D')}(\cdot)}dt 
\leq C'\int_0^{\pm\infty} \jap{t}^{-1-\m}dt<\infty.
\end{align*}
Combining these, we conclude \eqref{CK}, and completes the proof. \qed 

Now Theorem~1.1 follows from Theorem~5.1 and standard argument of the scattering theory. 

\begin{rem}
We note that the choice of the extension $V(x)$ on $\re^d$ is not unique, and hence $\F_\pm(t,\x)$ are not 
unique either. The modified wave operators $W_\pm^\F$ depend on such construction, but in general, 
they are equivalent up to multiplication by unitary Fourier multipliers from the right. Actually, if 
\[
W_\pm^\F=\slim_{t\to\pm\infty} e^{itH} e^{-i\F_\pm(t,D_x)} \quad \text{and}\quad
W_\pm^\G=\slim_{t\to\pm\infty} e^{itH} e^{-i\G_\pm(t,D_x)}
\]
exist, then the limit 
\[
G_\pm = \slim_{t\to\pm\infty} e^{i\G_\pm(t,D_x)} e^{-i\F_\pm(t,D_x)}
\]
exist. $G_\pm$ are obviously unitary Fourier multipliers, and $W_\pm^\F=W_\pm^\G G_\pm$. 
\end{rem}

\begin{rem}
If $V\in S^{-\m}(\ze^d)$ with $\m>1/2$, then we may employ a simpler approximate solution 
to the Hamilton-Jacobi equation. 
\[
U_D(t)=e^{-i\F^D(t,D_x)}, \quad \F^D(t,\x)= tp_0(\x) + \int_0^t V(sv(\x))ds, 
\]
is called the Dollard-type modifier, and we can show the existence of the modified wave operators 
similarly (see, e.g., \cite{RS} \S XI.9, \cite{DG} \S 4.9).
\end{rem}

\end{document}